\documentclass[11pt]{article}
\usepackage{amsthm}
\usepackage{amsmath}
\usepackage{amssymb}
\usepackage{refcount}
\usepackage{fullpage}
\usepackage{paralist}
\usepackage{appendix}
\usepackage{graphicx}
\usepackage{bm}
\usepackage[ruled,boxed,vlined]{algorithm2e}
\DontPrintSemicolon
\usepackage{verbatim}

\renewcommand{\phi}{\varphi}
\newcommand{\set}[1]{\left\{#1\right\}}
\newcommand{\stt}{\medspace | \medspace}

\newcommand{\eps}{\epsilon}

\newcommand{\given}{\medspace | \medspace}

\DeclareMathOperator{\argmin}{argmin}
\newcommand{\CONGEST}{\ensuremath{\mathsf{CONGEST}\ }}
\newcommand{\LOCAL}{\ensuremath{\mathsf{LOCAL}\ }}
\newcommand{\var}[1]{\mathit{#1}}
\newcommand{\children}{\var{children}}
\DeclareMathOperator*{\diam}{diam}

\newcommand{\hide}[1]{ }

\newcommand{\prob}[1]{\ensuremath{\text{\textsc{#1}}}}
\renewcommand{\mathbf}{\bm}

\theoremstyle{plain}
\newtheorem{theorem}{Theorem}[section]

\newtheorem{lemma}[theorem]{Lemma}
\newtheorem{property}[theorem]{Property}
\newtheorem{corollary}[theorem]{Corollary}

\newtheorem{definition}{Definition}
\newtheorem*{remark}{Remark}
\newtheorem{observation}{Observation}
\newtheorem{construction}[theorem]{Construction}
\renewcommand{\include}{\input}

\title{Distributed Property Testing for Subgraph-Freeness Revisited}
\author{Orr Fischer%
	\thanks{Computer Science Department, Tel-Aviv University. Email: orrfischer@mail.tau.ac.il}
	\and
	Tzlil Gonen%
	\thanks{Computer Science Department, Tel-Aviv University. Email: tzlilgon@tau.ac.il}
	\and
	Rotem Oshman%
	\thanks{Computer Science Department, Tel-Aviv University. Email: roshman@mail.tau.ac.il}
}

\begin{document}

\maketitle

\begin{abstract}
	In the subgraph-freeness problem, we are given a constant-size graph $H$, and wish to determine whether the network contains $H$ as a subgraph or not. The \emph{property-testing} relaxation of the problem only requires us to distinguish graphs that are $H$-free from graphs that are $\epsilon$-far from $H$-free, in the sense that an $\epsilon$-fraction of their edges must be removed to obtain an $H$-free graph. Recently, Censor-Hillel et.\ al.\ and Fraigniaud et.\ al.\ showed that in the property-testing regime it is possible to test $H$-freeness for any graph $H$ of size 4 in constant time, $O(1/\eps^2)$ rounds, regardless of the network size. However, Fraigniaud et. al. also showed that their techniques for graphs $H$ of size 4 cannot test $5$-cycle-freeness in constant time.

In this paper we revisit the subgraph-freeness problem and show that $5$-cycle-freeness, and indeed $H$-freeness for many other graphs $H$ comprising more than 4 vertices, can be tested in constant time. 
We show that $C_k$-freeness can be tested in $O(1/\epsilon)$ rounds for any cycle $C_k$, improving on the running time of $O(1/\epsilon^2)$ of the previous algorithms for triangle-freeness and $C_4$-freeness.
In the special case of triangles, we show that triangle-freeness can be solved in $O(1)$ rounds independently of $\eps$, when $\eps$ is not too small with respect to the number of nodes and edges.
We also show that $T$-freeness for any constant-size tree $T$ can be tested in $O(1)$ rounds, even without the property-testing relaxation. Building on these results, we define a general class of graphs for which we can test subgraph-freeness in $O(1/\epsilon)$ rounds. This class includes all graphs over 5 vertices except the 5-clique, $K_5$. 
For cliques $K_s$ over $s \geq 3$ nodes, we show that $K_s$-freeness can be tested in $O(m^{1/2-1/(s-2)}/\epsilon^{1/2+1/(s-2)})$ rounds, where $m$ is the number of edges in the graph.
Finally, we gives two lower bounds, showing that some dependence on $\eps$ is necessary when testing $H$-freeness for specific subgraphs $H$.
\end{abstract}




		\section{Introduction}
	

		The field of property testing asks the following question: given an input object $X$ and a property $\mathcal{P}$, can we distinguish the case where $X$ satisfies $\mathcal{P}$
		from the case where $X$ is \emph{$\eps$-far} from satisfying $\mathcal{P}$, in the sense that we would need to change an $\eps$-fraction of the bits in the representation of $X$ to obtain
		an object satisfying $\mathcal{P}$?
		This is a natural relaxation of the problem of \emph{exactly} whether $X$ satisfies a given property or not,
		and for hard problems, it can be much easier to solve than the exact version. 
		In this paper we study distributed property testing in the \CONGEST model, for the property of being \emph{$H$-free}, where $H$ is a fixed constant-size graph:
		we ask whether our network graph is $H$-free (that is, whether it does not contain $H$ as a subgraph),
		or whether we would need to remove many edges from the network graph to eliminate all copies of $H$.

		The subgraph-freeness problem has received significant attention in the distributed computing literature:
		the exact version was studied in ~\cite{DLP12,DFO14,CHKKLPJ15,G16},
		and the property-testing version was studied in~\cite{triangle_free} for triangles, and in~\cite{square-free} for graphs of size four.
		Of note, the exact version of subgraph-freeness is the only \emph{local} problem we are aware of which is known to be hard in the $\CONGEST$ model \cite{DFO14}:
		for example,
		in the \LOCAL model, where nodes can send as many bits as they want on each edge in a round,
		we can check if the graph contains a $k$-cycle in $O(k)$ rounds,
		but in the \CONGEST model, where the bandwidth on each edge is restricted,
		checking for odd-length cycles requires $\tilde\Theta(n)$ rounds (where $n$ is the number of nodes in the graph).
		
	The aim of the paper is to improve our understanding of the question: ``which types excluded subgraphs can be tested in constant time?''.
	We also explore several related questions, such as whether limiting the maximum degree in the graphs helps (by analogy to the bounded-degree model in sequential property testing),
	whether we can test $H$-freeness in \emph{sublinear} time for some subgraphs $H$ for which no constant-time algorithm is known,
	and whether there are cases where  we can test $H$-freeness with \emph{no dependence on the distance parameter $\epsilon$}, even when $\epsilon$ is sub-constant (e.g., $\eps = O(1/\sqrt{n})$).
	Using new ideas and combining them with previous techniques, we are able to extend and improve upon prior work, and point out some surprising answers to the questions above,
	which point to several aspects where distributed property testing for subgraph-freeness differs from the sequential analogue.

	\textbf{Our results.}
We begin by showing that for any size $k$ we can test $k$-cycle freeness in $O(1/\eps)$ rounds, improving on the running time of $O(1/\eps^2)$ for triangles and 4-cycles from~\cite{triangle_free,square-free}.
Next we show that for any tree $T$, we can test $T$-freeness exactly (without the property-testing relaxation) in constant time.
Both of the results extend to directed graphs in the directed version of the \CONGEST model.
Combining the two algorithms, we give a class of graphs $\mathcal{H}$ such that for any constant-sized $H \in \mathcal{H}$, we can test $H$-freeness in $O(1/\eps)$ rounds.
The class $\mathcal{H}$ consists of all graphs $H$ containing an edge $\set{u,v}$ such that each cycle in $H$ includes either $u$ or $v$ (or both).
This
includes all graphs of size 5 except for the 5-clique, $K_5$.

Next we turn our attention to the special case of cliques.
We present a different approach for detecting triangles, showing that when $\eps$ is not too small, we can \emph{eliminate}
the dependence on it in the running time: 
triangle-freeness can be tested in $O(1)$ rounds whenever $\epsilon \geq \min \set{ m^{-1/3}, n/m}$,
where $n$ is the number of nodes and $m$ is the number of edges.
We extend this approach to cliques of any size $s \geq 3$, and show that $K_s$-freeness can be tested in
$O\left(\epsilon^{-1/2-1/(s-2)}m^{1/2-1/(s-2)}\right)$ rounds. In particular, for constant $\eps$ and $s = 5$, we can test $K_5$-freeness in $O(m^{1/6})$ rounds.
We also modify the algorithm to work in \emph{constant time} in graphs whose maximum degree $\Delta$ is not too large with respect to the total number of edges, $\Delta = O( (\eps m)^{1/(s-2)})$.

Finally we consider the question of lower bounds.
We point out if we are not allowed to depend on the size and number of edges
in the graph, then a running time of $\Omega(1/\sqrt{\epsilon})$ is required for testing $C_k$-freeness for any $k \geq 4$.
We also exhibit a directed graph of size 4 which requires $\Omega(1/\eps)$ rounds to detect in the directed variant of the \CONGEST model.
And to conclude, we show that the \emph{Behrend graph}, the archtypical construction for showing lower bounds on subgraph-freeness in the sequential property testing world, which was also used in~\cite{square-free} to show a lower bound on one of their techniques, is probably not a hard case for $K_s$-detection, as it can be solved in a sub-polynomial number of rounds.

\subsection{Related Work}
 Property testing is an important notion in many areas of theoretical computer science, and has been used in wide-ranging contexts, from probabilistically-checkable proofs to coding and cryptography. The first paper to study property testing in graphs is~\cite{GGR98}, and much work followed; we refer to the surveys~\cite{Ron09,Fischer01,Gol98} for more background.
 Specifically, the problem of subgraph-freeness (also called \emph{excluded} or \emph{forbidden} subgraphs) has been extensively studied in the sequential property testing world ~\cite{Alon02,AlonFKS00,AKKR08,CGRSSS14}.
In their seminal work, \cite{AlonFKS00} showed that in the dense model, where the number of edges is $m = \Theta(n^2)$, $H$-freeness can be tested in $O(1)$ rounds for any fixed sized subgraph $H$, although
in some cases --- including triangles --- any solution independent of  $n$ must have a super-polynomial dependence on $\epsilon$.
In this sense, it is perhaps surprising that triangles turn out to be \emph{easy} for the distributed model, with a running time that does not even depend on $\epsilon$ unless $\epsilon$ is very small
compared to $n, m$.

Several recent works study \emph{distributed} property testing\cite{Brakerski2011,triangle_free,square-free}. Brakerski et. al \cite{Brakerski2011} studied the problem of detecting very large near-cliques
assuming that a large enough 
near-clique exists in the graph.
Censor-Hillel et al. \cite{triangle_free} formally introduced the question of distributed property testing, and showed that many sequential property testers can be imported to the distributed world;
they also showed that triangle-freeness can be tested in $O(1/\eps^2)$ rounds. Expanding upon their work, \cite{square-free} showed that testing $H$-freeness for any $4$-node graph $H$ can be done in $O(1/\eps^2)$
rounds, but they also showed that their techniques did not extend to 5-cycles (which we solve here) and 5-cliques (for which we are not able to give a constant-time algorithm, but do give a sublinear-time algorithm).



Some of our algorithms draw inspiration from a technique called \emph{color coding}, where we randomly color the nodes of the graph, and discard edges whose endpoints do not satisfy some condition on the colors.
This technique was introduced in~\cite{color_coding} and used there to detect cycles and path of fixed size $k$, and we use the technique in a similar way in Section~\ref{sec:cycles}.
\section{Preliminaries}
\label{sec:prelim}

We generally work with undirected graphs, unless indicated otherwise.
We let $N(v)$ as the neighbors of $v$, and $d(v)$ the degree of $v$.
We stress that throughout the paper, when we use the term \emph{subgraph}, we do not mean \emph{induced subgraph};
we say that $G' = (V', E')$ is a subgraph of $G = (V,E)$
if $V' \subseteq V, E' \subseteq E$.

We say that a graph $G = (V,E)$ is \emph{$\eps$-far from property $\mathcal{P}$} if at least $\eps |E|$ edges
need to be added to or removed from $E$ to obtain a graph satisfying $\mathcal{P}$.

The goal in distributed property testing for $H$-freeness is to solve the following problem:
if the network graph $G$ is $H$-free, then with probability $2/3$, all nodes should accept.
On the other hand, if $G$ is $\eps$-far from $H$-free, then with probability $2/3$, some node should reject.

We rely on the following fundamental property, which serves as the basis for most sequential property testers for $H$-freeness:
\begin{property}
	Let $G$ be $\eps$-far from being $H$-free, then $G$ has $\eps m/|E(H)|$ edge-disjoint copies of $H$.	
\end{property}

Our algorithms assign random colors to vertices of the graph,
and then look for a copy of the forbidden subgraph $H$ which received the ``correct colors''.
Formally we define:
\begin{definition}[Properly-colored subgraphs]
	Let $G = (V, E)$ and $H = ([k], F)$ be graphs, and let $G' = (V', E')$ be a subgraph of $G$
	that is isomorphic to $H$.
	We say that $G'$ is \emph{properly colored} with respect to a mapping $\var{color}_V : V \rightarrow [k]$ if
	there is an isomorphism $\varphi : V' \rightarrow [k]$ from $G'$ to $H$
	such that for each $u \in V'$ we have $\var{color}_V(v) = \varphi(v)$.
\end{definition}

\section{Detecting Constant-Size Cycles}
\label{sec:cycles}

In this section we show that $ C_k $-freeness can be tested in $O(1/\eps)$ rounds in the \CONGEST \\ model for any constant integer $ k>2 $.

\begin{theorem}\label{thm:cycles}
	For any constant $k > 2$,
	there is a 1-sided error distributed algorithm for testing $ C_k $-freeness
	which uses $O(1/\eps)$ rounds.
\end{theorem}

The key idea of the algorithm is to assign each node $u$ of the graph a random color $\var{color}(u) \in [k]$.
The node colors induce a coloring of both orientations of each edge, where $\var{color}( u,v) = (\var{color}(u), \var{color}(v))$.
We discard all edges that are not colored $(i, (i + 1) \bmod k)$ for some $i \in [k]$;
this eliminates all cycles of size less than $k$, while preserving a constant fraction of $k$-cycles with high probability.

Next, we look for a properly-colored $k$-cycle by choosing a random directed edge $(u_0, u_1)$%
\footnote{
	What we really want to do is choose a random node with probability proportional to its degree;
	choosing  random edge is a simple way to do that.
}
and carrying out a $k$-round \emph{color-coded BFS} from node $u_0$:
in each round $r = 0,\ldots,k-1$, the BFS only explores edges colored $(r, (r+1) \bmod k)$.
After $k$ rounds, if the BFS reaches node $u_0$ again, then we have found a $k$-cycle.

	Next we describe the implementation of the algorithm in more detail. We do not attempt to optimize the constants.
	To simplify the analysis, fix a set $\mathcal{C}$ of $\epsilon m/k$ edge-disjoint $k$-cycles (which we know exist if the graph is $\eps$-far from $C_k$-free). We abuse notation by also treating $\mathcal{C}$ as the set of edges participating in the cycles in $\mathcal{C}$.

	For the analysis, it is helpful to think of the algorithm as first choosing a random edge and then choosing random colors, and this is the way we describe it below.

	\paragraph*{Choosing a random edge}
	It is not possible to get all nodes of the graph to explicitly agree on a uniformly random directed edge in constant time (unless the graph has constant diameter),
	but we can emulate the effect as follows:
	each node $u \in V$ chooses a uniformly random weight $w( e ) \in [n^4]$ for each of its edges $e$.
	(Note that each edge has two weights, one for each of its orientations.)
	Implicitly, the directed edge we selected is the edge with the smallest weight in the graph, assuming that no two directed edges have the same weight.

	\begin{observation}
		With probability at least $1 - 1/n^2$, all weights in the graph are unique.
		\label{lemma:unique_weights}
	\end{observation}
	\begin{proof}
		For a given directed edge, the probability that another edge chooses the same weight is $1/n^4$;
		by union bound, the probability that there exists an edge that chose a weight shared with another edge
		is bounded by $1/n^2$.
	\end{proof}
	Let $\mathcal{E}_U$ be the event that all edge weights are unique.
	Conditioned on $\mathcal{E}_{U}$, the directed edge with the smallest weight is uniformly random.
	Let $e_0$ be this edge; implicitly, $e_0$ is the edge we select. (However, nodes do not initially know which edge was selected,
	or even if a single edge was selected.)

	Since the set $\mathcal{C}$ contains $\eps m/k$ edge-disjoint $k$-cycles, and the graph has a total of $m$ edges, we have:
	\begin{observation}
		We have $\Pr\left[ e_0 \in \mathcal{C} \given \mathcal{E}_U \right] = \eps$.
		\label{cor:random_edge}
	\end{observation}
	Let $\mathcal{E}_{Cyc}$ be the event that $e_0 \in \mathcal{C}$, and let $C_0 = \set{ u_0, u_1,\ldots, u_{k-1}}$ be the cycle to which $e_0$ belongs given $\mathcal{E}_C$, where $e_0 = (u_0, u_1)$.

	\paragraph*{Color coding.} 
	In order to eliminate cycles of length less than $k$, we assign to each node $u$ a uniformly random color $\var{color}(u) \in [k]$.
	Node $u$ then broadcasts $\var{color}(u)$ to its neighbors.

	Since the colors are independent of the edge weights, we have:
	\begin{observation}
		$\Pr\left[ \forall i \in [k] \medspace : \medspace \var{color}(u_i) = i \given \mathcal{E}_C, \mathcal{E}_U \right] 
		=
		\frac{1}{k^k}$.
	\end{observation}
	Let $\mathcal{E}_{Col}$ be the event that each $u_i$ received color $i$.
	Combining our observations above yields:
	\begin{corollary}
		$\Pr\left[ \mathcal{E}_U \cap \mathcal{E}_{Cyc} \cap \mathcal{E}_{Col} \right] > \frac{\eps}{2k^k}$.
	\end{corollary}

	Next we show thatn when $\mathcal{E}_U, \mathcal{E}_{Cyc}$ and $\mathcal{E}_{Col}$ all occur,
	we find a $k$-cycle.

	\paragraph*{Color-coded BFS}
	Each node $u$ stores the weight $\var{wgt}_u$ associated with the lightest edge it has heard of so far,
	and the root $\var{root}_u$ of the BFS tree to which it currently belongs.
	Initially,
	$\var{wgt}_u$ is set to the weight of the lightest of $u$'s outgoing edges,
	and $\var{root}_u$ is set to $u$.

	In each round $r = 0,\ldots,k-1$ of the BFS, 
	nodes $u$ with color $r$ send $(u, \var{wgt}_u, \var{root}_u)$ to their neighbors,
	and nodes $v$ with color $r + 1$ update their state:
	if they received a message $(u, \var{wgt}_u, \var{root}_u)$ from a neighbor $u$,
	they set $\var{wgt}_v$ to the lightest weight they received, and $\var{root}_v$
	to the root associated with that weight.

	After $k$ rounds, if some node colored $0$ receives a message $(v, \var{wgt}_v, \var{root}_v)$
	where $\var{root}_v = u$,
	then it has found a $k$-cycle, and it rejects.

	In Section~\ref{sec:complex}, we will use the same $C_k$-freeness algorithms, 
	but some nodes will be prohibited from taking certain colors.
	We incorporate this in Algorithm~\ref{alg:colorBFS} by having some nodes whose state is $\mathsf{abort}$.
	These nodes do not forward BFS messages and do not participate in the algorithm.

	\begin{lemma}
		If $\mathcal{E}_U, \mathcal{E}_{Cyc}$ and $\mathcal{E}_{Col}$ all occur,
		and if in addition the cycle $C_0$ contains no nodes whose state is $\mathsf{abort}$,
		then $u_0$ returns 1 and Algorithm~\ref{alg:colorBFS} finds a $k$-cycle (i.e., returns 1).
		\label{lemma:BFS}
		\label{lem:k_acyclic}
		\label{alg:C_k_detection}
	\end{lemma}
	\begin{proof}
		Let $C_0 = \set{ u_0,\ldots, u_{k-1}}$ where $e_0 = (u_0,u_1)$.
		We show induction on $r$ that at time $r < k - 1$,
		for each $s \leq r$,
		node $u_s$ has $\var{root}_{u_s}(r) = u_0$ and $\var{wgt}_{u_s}(r) = w(e_0)$.
		The base case is immediate, as $u_0$ initializes $\var{root}_{u_0}(0) \leftarrow u_0$
		and $\var{wgt}_{u_0}(0) \leftarrow w(e_0)$ (as $e_0$ is the lightest edge, and it is outgoing from $u_0$).

		For the step, assume the claim holds at time $r$, and consider time $r + 1$.
		In round $r + 1$, by I.H., node $u_r$ has $\var{root}_{u_r}(r) = u_0$
		and $\var{wgt}_{u_r}(r) = w(e_0)$.
		Since $\var{color}(u_r) = r$ (given $\mathcal{E}_{Col}$),
		node $u_r$ sends $(w(e_0), u_0)$ to its neighbors, including $u_{r+1}$.
		Since $e_0$ is the lightest edge, and $\var{color}(u_{r+1}) = r+1$, node $u_{r+1}$
		upon receiving $u_r$'s message
		sets $\var{wgt}_{u_{r+1}}(r+1) \leftarrow w(e_0), \var{root}_{u_{r+1}}(r+1) \leftarrow u_0$.
		The other nodes $u_0,\ldots,u_r$ do not change $\var{wgt}$ or $\var{root}$,
		as their color is not $r + 1$.

		At time $k - 1$, node $u_{k-1}$ 
		has $\var{root}_{u_{k-1}}(k-1) = u_0$.
		Thus, in round $k$, it sends $(w(e_0), u_0)$ back to node $u_0$,
		which then returns 1.

	\end{proof}

	\begin{proof}[Proof of Theorem~\ref{thm:cycles}]
		Suppose that $G$ is $\eps$-far from $C_k$-free.
		We have no nodes whose state is $\mathsf{abort}$ (as we said, the $\mathsf{abort}$ state will be used in Section~\ref{sec:complex}).
		Each time we draw random colors and weights in Alg.~\ref{alg:Ck}, the probability that 
		$\mathcal{E}_U, \mathcal{E}_{Cyc}$ and $\mathcal{E}_{Col}$ all occur
		is at least $\frac{\eps}{2k^k}$;
		therefore, the probability that we fail to detect a $k$-cycle after $\lceil 20k^k / \eps \rceil$
		attempts is at most $1/10$.
	\end{proof}

	\begin{algorithm}[H]
		$\var{root} \leftarrow u$\;
		$\var{wgt} \leftarrow \min \set{ w(u,v) \stt v \in N(u)}$\;
		\For{$r = 0,\ldots,k-1$}
		{
			\lIf{$\var{color} = r$ and $\var{state} \neq \mathsf{abort}$}
			{
				send $(\var{wgt}, \var{root})$ to neighbors
			}
			receive $(w_1, r_1),\ldots,(w_t, r_t)$ from neighbors\;
			\If{$\var{color} = (r + 1) \bmod k$}
			{
				$i \leftarrow \argmin \set{ w_1,\ldots,w_t }$\;
				\If{ $w_i < \var{wgt}$ }
				{
					$\var{root} \leftarrow r_i, \var{wgt} \leftarrow w_i$ \;
				}
			}
			\lIf{ $r = k-1$ and $u \in \set{ r_1,\ldots,r_t}$}
			{
				\Return 1
			}
		}
		\Return 0\;
		\caption{Procedure \texttt{ColorCodedBFS}, code for node $u$}
		\label{alg:colorBFS}
	\end{algorithm}

	\begin{algorithm}[H]
		\For{$i = 1,\ldots,\lceil 20k^k / \eps \rceil$}
		{
			$\var{color} \leftarrow$ uniformly random color from $[k]$\;
			\ForEach{$v \in N(u)$}
			{
				$w(u,v) \leftarrow $ uniformly random weight from $\set{0,\ldots,n^4-1}$\;
			}
			$\var{res} \leftarrow$ \texttt{ColorCodedBFS}()\;
			\lIf{$\var{res} = 1$}
			{
				\textbf{reject}
			}
		}
		\textbf{accept}\;
		\caption{$C_k$-freeness algorithm, code for node $u$}
		\label{alg:Ck}
	\end{algorithm}


	\section{Detecting Constant-Size Trees}
	\label{sec:trees}
			\label{lem:tree_detection}
		In this section we show that for any constant-size tree $T$,
		we can test $T$-freeness \emph{exactly} (that is, without the property-testing relaxation) in $O(1)$ rounds.
		Let the nodes of $T$ be $0,\ldots,k-1$.
		We arbitrarily assign node $0$ to be the root of $T$, and orient the edges of the tree upwards toward node $0$.
		Let $R$ be the depth of the tree, that is, the maximum number of hops from any leaf of $T$ to node $0$.
		Finally, let $\children(x)$ be the children of node $x$ in the tree.

		In the algorithm, we map each node of the network graph $G$
		onto a random node of $T$ by assigning it a random color from $[k]$.
		Then we check if there is a copy of $T$ in $G$ that was mapped ``correctly'',
		with each node receiving the color of the vertex in $T$ it corresponds to.
		
		Initially the state of each node is ``open'' if it is an inner node of $T$, and ``closed'' if it is a leaf.
		The algorithm has $R$ rounds, in each of which all nodes broadcast their state and their color to their neighbors.
		When a node with color $j$ hears ``closed'' messages from nodes with colors matching all the children of node $j$ in $T$,
		it changes its status to ``closed''.
		After $R$ rounds, if node $0$'s state is ``closed'', we reject.

		\begin{algorithm}[H]
			\eIf{ $\children(\var{color}) = \emptyset$}
			{
				$\var{state} \leftarrow \mathsf{closed}$\;
			}
			{
				$\var{state} \leftarrow \mathsf{open}$\;
			}
			$\var{missing} \leftarrow \children(\var{color})$\;
			\For{$r = 1,\ldots,R$}
			{
				send $(\var{color},\var{state})$ to neighbors \;
				receive $(c_1,s_1),\ldots,(c_t, s_t)$ from neighbors\;
				\ForEach{$i = 1,\ldots,t$}
				{
					\If{$c_i \in \var{missing}$ and $s_i = \mathsf{closed}$}
					{
						$\var{missing} \leftarrow \var{missing} \setminus \set{ c_i }$\;
						\lIf{$\var{missing} = \emptyset$}
						{
							$\var{state} \leftarrow \mathsf{closed}$
						}
					}
				}
			}
			\eIf{$\var{color} = 0$ and $\var{state} = \mathsf{closed}$}
			{
				\Return 1\;
			}
			{
				\Return 0\;
			}

			\caption{Procedure \texttt{CheckTree}, code for node $u$}
			\label{alg:CheckTree}
		\end{algorithm}
	
		\begin{algorithm}[H]
			\For{$i = 1,\ldots, 10 k^{k}$}
			{
				$\var{color} \leftarrow$ random color from $\set{0,\ldots,k-1}$\;
				$\var{res} \leftarrow$ \texttt{CheckTree}()\;
				\lIf{$\var{res} = 1$}
				{
					\textbf{reject}
				}
			}
			\textbf{accept}\;

			\caption{$T$-detection algorithm, code for node $u$}
			\label{alg:tree}
			\label{alg:T_k_detection}
		\end{algorithm}

	Let $x \in \set{0,\ldots,k-1}$ be a node of $T$, let $T'$ be the sub-tree rooted at $x$,
	and let $G' = (U, E')$ be a subgraph of $G = (V,E)$ isomorphic to $T'$.
	We say that $G'$ is \emph{properly colored} if there is an isomorphism $\varphi$ from $G'$ to $T'$,
	such that $\var{color}(u) = \varphi(u)$.
	(There may be more than one possible isomorphism from $G'$ to $T'$.)

	\begin{lemma}
		Let $u$ be a node with color $\var{color}(u) = x$,
		and let $T'$ be the sub-tree of $T$ rooted at $x$.
		Let $h_x$ be the height of $x$, that is, the length of the longest path from a leaf of $T'$ to $x$.
		Then at any time $t \geq h_x$ in the execution of Algorithm~\ref{alg:tree},
		we have $\var{state}_u(t) = \mathsf{closed}$ iff there 
		is a subgraph $G'$ containing $u$, which is isomorphic to $T'$ and properly colored.
		\label{lemma:tree_ind}
	\end{lemma}
	\begin{proof}
		By induction on $h_x$.
		Since nodes never change their status from $\mathsf{closed}$ back to $\mathsf{open}$,
		it suffices to show that at time $t = h_x$ we have $\var{state}_u(t) = \mathsf{closed}$.

		For the leafs of $T$ (which have height 0) the claim is immediate.
		Now suppose that the claim holds for all nodes at height $h$, and let $x$ be a node at height $h_x = h+1$.
		Let $G'$ be the subgraph containing $u$ isomorphic to $T'$, and let $\varphi$ be the isomorphism from $G'$ to $T'$
		with respect to which $G'$ is properly colored.
		Finally, let $v_1,\ldots,v_{\ell}$ be the nodes of $G'$ mapped by $\varphi$ to the children of $x$ in $T'$.
		The height of $x$'s children is at most $h$,
		so by the induction hypothesis, at time $h$ we have $\var{state}_{v_i}(h) = \mathsf{closed}$
		for each $i = 1,\ldots,\ell$.
		Thus, no later than round $h$, node $u$ receives messages $(\var{color}(v_i), \mathsf{closed})$
		from each $v_i$, emptying out $\children_u$
		and setting $\var{state}_u$ to $\mathsf{closed}$.
	\end{proof}

	\begin{corollary}
		For any node $u \in V$,
		at time $h$ we have $\mathsf{state}_u(h) = \mathsf{closed}$ iff $u$ is the root of a properly-colored copy of $T$.
		\label{cor:root}
	\end{corollary}

	\begin{corollary}
		If $G$ contains a copy of $T$,
		then Algorithm~\ref{alg:tree} return 1 with probability $9/10$.
	\end{corollary}
	\begin{proof}
		Fix a subgraph $G'$ which is isomorphic to $T$.
		Each time we pick a random coloring, the probability that $G'$ is properly colored is at least $1/k^k$
		(perhaps more, if there is more than one isomorphism mapping the nodes of $G'$ to $T$).
		By Corollary~\ref{cor:root}, if $G'$ is properly colored, the root of the tree will discover this
		and return $1$.
		Therefore the probability that we fail $\lceil 10 k^k \rceil$ times is at most $1/10$.
	\end{proof}

	\section{Detecting Constant-Size Complex Graphs}
	\label{sec:complex}
	\label{alg:h_detection}
	
	In this section we define a class $\mathcal{H}$ of graphs, and give an algorithm for detecting those graphs in constant number of round	(taking the size of the graph as a constant).
	The class $\mathcal{H}$ includes all graphs of size 5 except $K_5$ (see subsection \ref{subsect:five_nodes}).
	
	\paragraph*{Definition of the class $\mathcal{H}$}
		\label{construction}
		The class $\mathcal{H}$ contains all graphs that have the following property: there exists an edge $(u,v)$ such that any cycle in the graph contain at least one of $u$ and $v$. Equivalently, the class $\mathcal{H}$ contains all connected graphs that can be constructed using the following procedure:
	\begin{enumerate}
		\item We start with two nodes, 0 and 1, with an edge between them
		\item Add any number of cycles $C_1,\ldots,C_{\ell}$ using new nodes, such that:
		\begin{itemize}\label{stage:cycles}
			\item Each cycle $C_i$ contains either node $0$ or node $1$ or both; and
			\item With the exception of nodes $0,1$, the cycles are node-disjoint.
		\end{itemize}
		\item Select a subset $R$
		of the nodes added so far, and for each node $x$ selected, attach a tree $T_x$ rooted
			at $x$ using ``fresh'' nodes (that is, with the exception of node $x$,
			each tree $T_x$ that we attach is node-disjoint from the graph
			constructed so far, including trees $T_y$ added for other nodes $y \neq x$).
			\label{stage:trees}
		\item For each $x \in \set{0,1}$, add edges $E_x$ between node $x$ and some subset of nodes
			added in the previous steps.
			\label{stage:edges}
	\end{enumerate}

\begin{lemma}\label{lem:equivalence}
	The two definitions of $\mathcal{H}$ are equivalent.
\end{lemma}
\begin{proof}
	Clearly because in family of construction every cycle must pass through the vertices $0$ or $1$ it follows that it's contained in the other definition.
	
	The other direction is shown by proving in induction on the number of edges. The induction's hypothesis is that for any $G$ in the other definition there exists a construction recipe $S$ that constructs $G$, in which $u,v$ is mapped to $0,1$ respectively.
	
	The base of the induction is trivial. Let $G$ be a graph in $\mathcal{H}$ with $a$ edges. If there's a vertex $w$ such that $d(w) = 1$, then the edge connecting it to the rest of the graph isn't in a cycle. By removing the edge and applying the induction assumption we get a construction a recipe $S$. The graph can be constructed by applying the same construction as $S$ and adding the edge of $w$ in the third stage.
	
	Otherwise, since the minimal degree is $2$, all nodes must participate in a cycle. Recall that by definition, all cycles must pass through either $u$ or $v$. Consider a neighbor of $u$, $w \neq v$. $w$ must have a path to $u$ that doesn't contain $(u,w)$. By Disconnecting $(u,w)$ and applying induction assumption, we get recipe $S$.  The graph can be constructed by applying the same construction as $S$ and adding the edge of $w$ in the forth stage.
\end{proof}

	Our algorithm for testing $H$-freeness for $H \in \mathcal{H}$ combines the ideas from the previous sections.
	We begin by color-coding the nodes of $G$, mapping each node onto a random node of $H$.
	Next, we choose a random directed edge $(u_0,u_1)$ from among the edges mapped to $(0,1)$, and begin the task of verifying that 
	the various components of $H$ are present and attached properly.

	For simplicity, below we describe the verification process assuming that we really do choose a unique random edge,
	and all nodes know what it is;
	however, we cannot really do this, so we substitute using random edge weights as in Section~\ref{sec:cycles}.
	\begin{enumerate}[(I)]
		\item Nodes $u_0$ and $u_1$ broadcast the chosen edge $(u_0, u_1)$ for $\diam(H)$ rounds.
		\item Any node whose color is 0 or 1 but which is not $u_0$ or $u_1$ (resp.)
			sets its state to $\mathsf{abort}$.
		\item For each edge $\set{ b, x} \in E_b$, where $b \in \set{0,1}$, 
			nodes colored $x$ verify that they have an edge to node $u_b$;
			if they do not, they set their state to $\mathsf{abort}$.
		\item For each tree $T_x$ added in stage~\ref{stage:trees} of the construction,
			we verify that a properly-colored copy of $T_x$ is present,
			by having nodes colored $x$ call Algorithm~\ref{alg:CheckTree},
			with the colors replaced by the names of the nodes in $T_x$.
			We denote this by \texttt{CheckTree}$(T_x)$.

			If a node colored $x$ fails to detect a copy of $T_x$ for which it is the root,
			it sets its state to $\mathsf{abort}$ for the rest of the current attempt.
		\item 
			For each $i = 1,\ldots,\ell$, we test for a properly-colored copy of $C_i$.
			We define the \emph{owner} of $C_i$, denoted $\var{owner}(C_i)$,
			to be node $0$ if $C_0$ contains $0$, and otherwise node 1.
			To verify the presence of $C_i$,
			we call Algorithm~\ref{alg:colorBFS}, using the names of the nodes in $C_i$ as colors:
			instead of color 0 we use $\var{owner}(C_i)$,
			and the remaining colors are mapped to the other nodes of $C_i$ in order (in a arbitrary orientation of $C_i$).
			We denote this call by \texttt{ColorBFS}$(C_i)$.
			(As indicated in Alg.~\ref{alg:colorBFS}, nodes whose state is $\mathsf{abort}$ do not participate.)
		\item If both $u_0$ and $u_1$ are not in state $\mathsf{abort}$, $u_0$ rejects, otherwise it accepts. All other nodes accept.
	\end{enumerate}

\begin{algorithm}[H]
	\eIf{$\var{color} = 0$}
	{
		$\var{u0} \leftarrow u$\;
		$\var{wgt} \leftarrow \min \set{ w(u,v) \stt v \in N(u), \var{color}(v) = 1 }$\;
		$\var{u1} \leftarrow \argmin \set{ w(u,v) \stt v \in N(u), \var{color}(v) = 1 }$\;
	}
	{
		$\var{u0} \leftarrow \bot$,
		$\var{wgt} \leftarrow \infty$,
		$\var{u1} \leftarrow \bot$\;
	}
	\For{$r = 1,\ldots,\diam(H)$}
	{
		send $(\var{u0}, \var{u1}, \var{wgt})$ to neighbors\;
		receive $(\var{v0}_1,\var{v1}_1,\var{w}_1),\ldots,(\var{v0}_t,\var{v1}_t,\var{w}_t)$\;
		$i \leftarrow \argmin \set{ \var{w}_1,\ldots, \var{w}_t }$\;
		\lIf{ $\var{w}_i < \var{wgt}$ }
		{
			$(\var{u0}, \var{u1}) \leftarrow (\var{v0}_i, \var{v1}_i)$
		}
	}
	$\var{state} \leftarrow \mathsf{OK}$\;
	\lIf{ ($\var{color} = 0$ and $\var{u0} \neq u$) or $\var{color} = 1$ and $\var{u1} \neq u$}
	{
		$\var{state} \leftarrow \mathsf{abort}$
	}
	\lIf{ ($\set{ 0, \var{color}} \in E_0$ and $u0 \not \in N(u)$)
		or ($\set{ 1, \var{color}} \in E_1$ and $u1 \not \in N(u)$)}
	{
		$\var{state} \leftarrow \mathsf{abort}$
	}
	\ForEach{ $x \in R$}
	{
		$\var{res} \leftarrow $ \texttt{CheckTree}$(T_x)$\;
		\lIf{ $\var{color} = x$ and $\var{res} \neq 1$}
		{
			$\var{state} \leftarrow \mathsf{abort}$
		}
	}
	\For{$ i = 1,\ldots, \ell$}
	{
		$\var{res} \leftarrow$ \texttt{ColorBFS}($C_i$)\;
		\lIf{$\var{color} = \var{owner}(C_i)$ and $\var{res} \neq 1$}
		{
			$\var{state} \leftarrow \mathsf{abort}$
		}
	}
	\eIf{$u = u_0$ and $\var{state} \neq \mathsf{abort}$ and $u_1$ with state $\neq \mathsf{abort}$}
	{
		\Return 1\;
	}
	{
		\Return 0\;
	}
	\caption{\texttt{CheckH}, code for node $u$}
	\label{alg:CheckH}
\end{algorithm}

\paragraph*{Analysis}
Fix a set $\mathcal{S}$ of $\eps m/|E(H)|$ edge-disjoint copies of $H$ in $G$,
and let $E_{\mathcal{S}}$ be the set of edges participating in these copies ($|E_H| = \eps m$).
When we choose a random directed edge, the probability that the edge is in $E_{\mathcal{S}}$ is at least $\eps  m / m = \eps $.
Since the edge weights are independent of the colors,
given this event, the probability that the copy we hit is properly colored is at least $1/k^k$;
therefore, the overall probability that we hit a properly colored copy is at least $\eps / k^k$.
Let $\mathcal{E}$ be this event, and let $G'$ be the properly-colored copy we hit (note that $G'$ is unique, because we restricted attention to the subgraphs in $\mathcal{S}$, which are edge-disjoint, and $G'$ contains the edge $(u_0, u_1)$).

Conditioned on $\mathcal{E}$,
Cor.~\ref{cor:root} shows that for each $x \in R$, node $x$ returns 1 when we call \texttt{CheckTree}$(T_x)$;
in addition, the verification of edges in $E_0$ and $E_1$ succeeds, as these edges are present and colored correctly.
Therefore no nodes of $G'$ set their state to $\mathsf{abort}$ in these steps.
Thus, by Lemma~\ref{lemma:BFS}, for each cycle $C_i$, the owner of the cycle returns 1 when we call \texttt{ColorBFS}$(C_i)$,
and

\hide{
\paragraph{Examples}

\begin{figure}
\begin{minipage}[b]{.4\linewidth}
	\centering
\begin{pspicture}(6,4)
	\psframe(0,0)(6,4)
\cnodeput(3,2){A}{0}
\cnodeput(5,3){B}{1}
\cnodeput(5,1){C}{2}
\cnodeput(1,3){D}{3}
\cnodeput(1,1){E}{4}
\ncline{B}{C}
\ncline{A}{C}
\ncline{A}{D}
\ncline{D}{E}
\ncline{A}{E}
\psset{linewidth=0.1}
\ncline{A}{B}
\end{pspicture}
\subcaption{The butterfly graph: two cycles added in Stage~\ref{stage:cycles}}\label{fig:1a}
\end{minipage}%
\hspace{.1\linewidth}
\begin{minipage}[b]{.4\linewidth}
	\centering
\begin{pspicture}(6,4)
	\psframe(0,0)(6,4)
	\SpecialCoor
	\rput(3,1.85){
\cnodeput(1.5;90){U4}{4}
\cnodeput(1.5;162){U0}{0}
\cnodeput(1.5;234){U1}{1}
\cnodeput(1.5;306){U2}{2}
\cnodeput(1.5;378){U3}{3}
}
\ncline{U1}{U2}
\ncline{U2}{U3}
\ncline{U3}{U4}
\ncline{U4}{U0}
\ncline{U0}{U2}
\ncline{U0}{U3}
\ncline{U1}{U3}
\ncline{U1}{U4}
\psset{linewidth=0.1}
\ncline{U0}{U1}
\end{pspicture}
\subcaption{$K_5$ minus one edge: one cycle added in Stage~\ref{stage:cycles}, four edges added in Stage~\ref{stage:edges}}\label{fig:1b}
\end{minipage}%
\end{figure}
}

\subsection{List of $5$-node connected graphs}
\label{subsect:five_nodes}
To show that indeed the algorithm \ref{alg:CheckH} detects any $5$ node connected graph, excluding $K_5$, we include a full list of all the connected graphs on $5$ nodes (up to isomorphism), and label the nodes, where the $0,1$ vertices correspond to the $0,1$ nodes in algorithm \ref{alg:CheckH}. We note that indeed in all these subgraphs, all the cycles pass through either $0$ or $1$.
\begin{figure}[h]
	\includegraphics[width=0.70\textwidth]{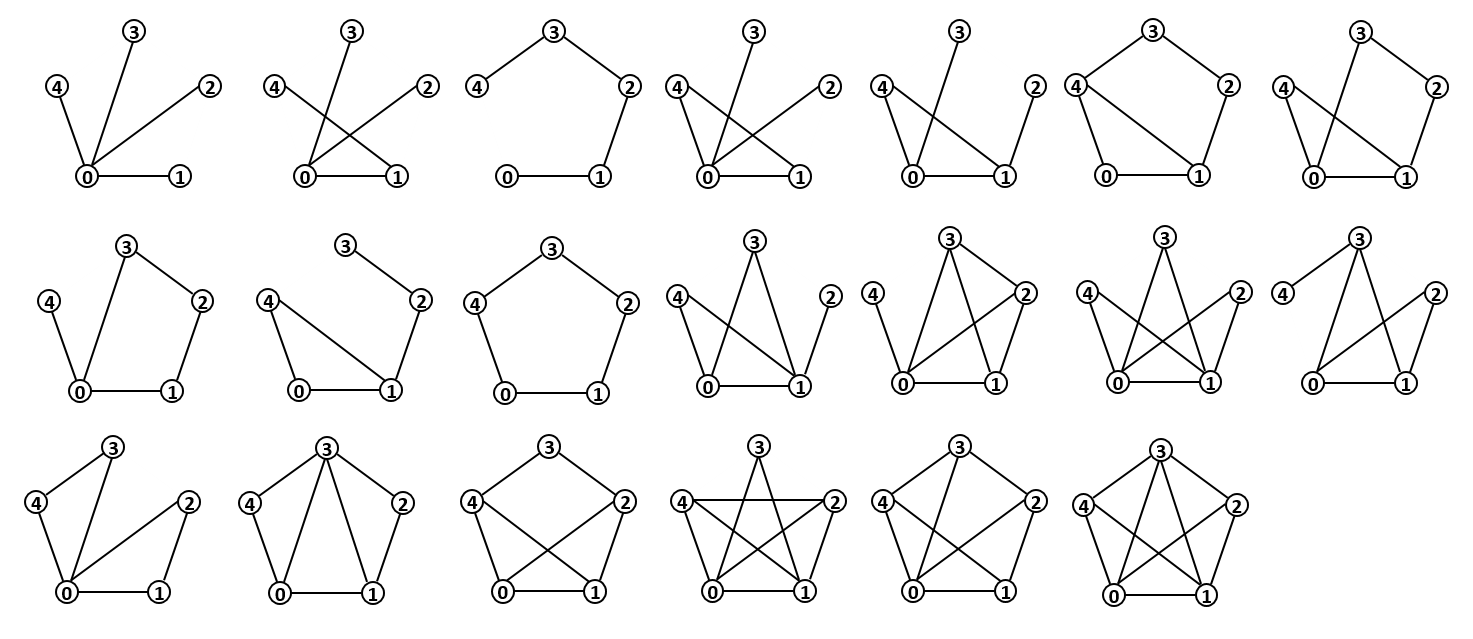}
	\caption{List of the $5$-node connected graphs, excluding $K_5$ (up to isomorphism)}
\end{figure}

\section{Testing $K_s$-Freeness}
\label{sec:cliques}
In previous sections it was shown how to test $K_3$ and $K_4$ freeness in $O(1/\eps)$ rounds of communication. In this section we describe how to test $K_s$-freeness for cliques of any constant size $s$, in a sublinear number of rounds. Moreover, we show that triangle-freeness can be tested in $O(1)$ rounds, with no dependence on $\eps$, when $\min\left( {\frac{n}{m},m^{-1/3}}\right) \leq \eps\leq 1$. Finally, we show that if the maximal degree is bounded by $O((\eps m)^{\frac{1}{s-2}})$ then $K_s$-freeness can be tested in $O(1)$ rounds.

\subsection{Algorithm overview}

The basic idea is the following simple observation: suppose that each node $u$ could learn the entire subgraph induced by $N(u)$,
that is, node $u$ knew for any two $v_1, v_2 \in N(u)$ whether they are neighbors or not.
Then $u$ could check if there is a set of $s$ neighbors in $N(u)$ that are all connected to each other,
and thus know if it participates in an $s$-clique or not.
How can we leverage this observation?

For nodes $u$ with high degree, we cannot afford to have $u$ learn the entire subgraph induced by $N(u)$,
as this requires of $N(u)^2$ bits of information.
But fortunately, if $G$ is $\eps$-far from $K_s$-free, then there are many copies of $K_s$
that contain some fairly low-degree nodes, as observed in~\cite{square-free}:
	
\begin{lemma}[\cite{square-free}]
	Let $I(G)$ be the set of edges in some maximum 
	set of edge-disjoint copies of $H$, and let $g(G) = \set{ (u,v) \medspace | \medspace d(u)d(v) \leq 2m|E(H)|/\epsilon}$.
	Then $|I(G) \cap g(G)| \geq \epsilon m / (4|E(H)|)$. 
	\label{lemma:good_edges}
\end{lemma}
\begin{remark}
	\cite{square-free} considers only subgraphs $H$ with $4$ vertices and constant $\eps$, but their proof works for any subgraph $H$
	and any $0 < \eps \leq 1$.
\end{remark}

The focus in~\cite{square-free} is on \emph{good edges}, which are edges satisfying the condition in Lemma~\ref{lemma:good_edges},
but here we need to focus on the \emph{endpoints} of such edges.
We call $u \in V$ a $\emph{good vertex}$ if its degree is at most $\sqrt{2m|E(H)| /  \epsilon}$,
and we say that a copy of $H$ in $G$ is a \emph{good copy} if it contains a good vertex.
Since each copy of $H$ in $I(G)$ contributes at most $|E(H)|$ edges to $g(G)$,
\begin{corollary}\label{coro:good_copies}
	If $G$ is $\eps$-far from $H$-free, then $G$ contains at least $\epsilon m / (4|E(H)|^2)$ edge-disjoint good copies of $H$. 
\end{corollary}

Because there are many good edge-disjoint copies of $K_s$,
we can \emph{sparsify} the graph
and still retain at least one good copy of $K_s$.

We partition $G$ into many edge-disjoint sparse subgraphs, by having each vertex $u$ choose for each neighbor $v \in N(u)$
a random color $\var{color}(v) \in \set{1,\ldots,C(u)}$,
	where the size of the color range, $C(u)$, will be fixed later.
	This induces a partition of $G$'s edges into $C(v)$ color classes;
	let $N_c(u)$ denote the set of neighbors  $v \in N(u)$ with $\var{color}(v) = c$.
	The expected size of $N_c(u)$ is $d(v) / C(v)$.

	With this partition in place, we begin by showing how to solve triangle-freeness in constant time,
	and then extend the algorithm to other cliques $K_s$ with $s > 3$.

\subsection{\boldmath Testing triangle-freeness for $\eps \in [ \min \set{ m^{-1/3}, n/m}, 1]$ in $O(1)$ rounds}

Assume that $\eps$ is not too small with respect to $n$ and $m$: $\eps \geq \min \set{ m^{-1/3}, n/m}$.
Then we can improve the algorithm from Section~\ref{sec:cycles} and test triangle-freeness in \emph{constant} time that does not depend on $\eps$.


To test triangle-freeness, we set $C(v) = \lceil d(v)/200 \rceil$.
Each node chooses a random color for each neighbor from the range $\set{1,\ldots,C(v)}$.
Then, we go through the color classes $c = 1,\ldots,C(v)$ \emph{in parallel},
and for each color class $c$, we look for a triangle containing two edges from $N_c(u)$:
let $N_c(u) = \set{ v_1,\ldots,v_{t_c}}$.
for $R = 202e^2$ rounds $r = 1,\ldots,R$,
node $u$ sends $v_r$ to all neighbors $v_1,\ldots,v_r$ in $N_c(u)$,
and each neighbor $v_i$ responds by telling $u$ whether it is also connected to $v_r$,
that is, whether $v_r \in N(v_i)$ (note that we do not insist on the edge $(v_r, v_i)$ having color $c$).
If $v_r \in N(v_i)$, then node $u$ has found a triangle, and it rejects.
If after $202e^2$ attempts node $u$ has not found a triangle in any color class, it accepts.


\begin{lemma}\label{lem:triangle-detection}
	If $G$ is $\epsilon$-far from $K_s$-free, then with probability $2/3$, at least one vertex detects a triangle.
\end{lemma}
\begin{proof}
	Let $\mathcal{T}$ be a set of edge-disjoint good triangles in $G$,
	of size $|\mathcal{T}|\geq \epsilon m / (4|E(T)|^2) = \epsilon m / 36$.
	By Corollary~\ref{coro:good_copies} we know that there is such a set.

	Assume that $\mathcal{T} = \set{T_1,\ldots,T_t}$.
	By definition, each good triangle has a \emph{good vertex}; let $v_i$ be a good vertex from the $i$'th triangle $T_i$.

	For each $i \in \set{1,\ldots,t}$,
	let $A_i$ be the event that $v_i$ assigned the same color, $c_i$, to the other two vertices in $T_i$,
	and let $X_i$ be an indicator for $A_i$.
	We have $\Pr\left[ X_i = 1 \right] = 1/C(v_i) = 200/d(v_i)$.
	Also, since the triangles in $\mathcal{T}$ are edge-disjoint, $X_1,\ldots,X_t$ are independent.
	Now let $X = \sum_{i = 1}^t X_i$ be their sum;
	then
	\begin{equation*}
		\Pr[ X = 0]
		=
		\Pr[\bigcap_{i=1}^{t} (X_i = 0)]
		= 
		\prod_{i=1}^{t} \left( 1-\frac{1}{C(v_i)}\right)
		= 
		\prod_{i=1}^{t}  \left( 1-\frac{200}{d(v_i)}\right).
	\end{equation*}
	We divide into two cases:
	\begin{enumerate}[I.]
		\item 
			$m < n^{3/2}$: then $\eps \geq \min\left( {\frac{n}{m},m^{-1/3}}\right) = m^{-1/3}$. Recall $v_i$ is a \emph{good vertex}, which means $d(v_i)\leq\sqrt{6m/\eps}$, and therefore
	\begin{multline*}
		\prod_{i=1}^{t}  \left( 1-\frac{200}{d(v_i)}\right)
		\leq
		\left( 1-\frac{200}{\sqrt{6m/\eps}}\right)^{t}
		\leq
		e^{-\frac{200t}{\sqrt{6m/\eps}}}
		\leq
		e^{-\frac{200\eps m}{36\sqrt{6m/\eps}}}
		=
    		e^{-\frac{200\eps^{3/2} \sqrt{m}}{36\sqrt{6}}}\leq e^{-2}
		.
   	\end{multline*}
\item 
	$m \geq n^{3/2}$: then $\eps \geq \min\left( {\frac{n}{m},m^{-1/3}}\right) = \frac{n}{m}$.
	The degree of each vertex is no more then $n$, and hence
   	\begin{align*}
   		\prod_{i=1}^{t}  \left( 1-\frac{200}{d(v_i)}\right)
		&
		\leq\left( 1-\frac{200}{n}\right)^{t}
		\leq
    		e^{-\frac{200t}{n}} \leq e^{-\frac{200\eps m}{36n}}
    		\leq e^{-2}.
	\end{align*}
	\end{enumerate}
	So in any case we get
	$Pr[X = 0]\leq e^{-2}$.
		
	Conditioned on $X \geq 1$, there is at least one vertex $v_i$
	which put two of its triangle neighbors in the same color class $c_i$,
	which means that if $N_c(v_i)$ is no larger than $200e^2$, node $v_i$ will go through all neighbors in $N_c(v_i)$ and
	find the triangle.
	Because the colors of the edges are independent of each other, conditioning on $A_i$ does not change the expected size of $N_{c_i}(v_i)$ by much:
	we know that the other two vertices in $T_i$ received color $c_i$, but the remaining neighbors are assigned to a color class independently.
	The expected size of $|N_{c_i}(v_i)|$ is therefore $(d(v_i)-2) / C(v_i) + 2 < 202 = R / e^2$,
	and by Markov, $\Pr\left[ |N_{c_i}(v_i)| > R  \right] \leq 1/e^2$.

	To conclude, by union bound, the probability that no node $v_i$ has $X_i = 1$,
	or that the smallest node $v_i$ with $X_i = 1$ has $|N_c(v_j)| > 200e^2$ for the smallest color class $c$ containing
	two triangle neighbors,
	is at most $1/e^2 + 1/e^2 < 1/3$.
\end{proof}

\begin{algorithm}
	$C(u) \gets \lceil d(u)/200 \rceil, R \gets \lceil 200e^2 \rceil$ \;
	Choose a random color $\var{color}(v) \in \set{1,\ldots,C(u)}$ for each $v \in N(u)$\;
	\ForEach{$c = 1,\ldots,C(u)$}
	{
		$N_c \leftarrow \set{ v \in N(u) : \var{color}(v) = c}$\;
		$\var{candidates}_c \leftarrow N_c$\;
	}
	\For{$r=1,...,R$}
	{
		\ForEach{$c=1,...C(u)$ (in parallel)}
		{
			$v \leftarrow \min \var{candidates}_c$\;
			$\var{candidates}_c \leftarrow \var{candidates}_c \setminus \set{ \min \var{candidates}_c}$\;
			query each neighbor $w \in N_c$ to ask if $v \in N(w)$\;
			\lIf{$\exists w \in N_c : v \in N(w)$}
			{
				\textbf{reject}
			}
		}
	}
	\textbf{accept}\;
\caption{Triangle detection: code for node $v$}
\label{alg:triangle_detection}
\end{algorithm}

\subsection{General tester for $K_s$-freeness}

Use the same algorithm but with a different setting of the parameters,
we can test $K_s$-freeness for any $s \geq 3$.

\begin{theorem}\label{thm:cliques}
There is a 1-sided error distributed property-testing algorithm for $ K_s$-freeness, for any constant $ s \geq 4$, with running time $O(\eps^{\frac{-s}{2(s-2)}}m^{\frac{s-4}{2(s-2)}})$.
\end{theorem}
\begin{corollary}
There is a 1-sided error distributed property-testing algorithm for $K_5$-freeness, with running time $O(m^{1/6})$.
\end{corollary}

We set
\begin{equation*}
	C(u) = \left\lceil{\left( \frac{1}{2s^4}\eps m\right) }^{\frac{1}{s-2}} \right\rceil
\end{equation*}
to be the number of color classes at node $u$,
and
\begin{equation*}
	R = 2s^4 e^2 \left[ \eps^{-1/2-1/(s-2)}m^{1/2-1/(s-2)} + s - 1 \right]
\end{equation*}
to be the timeout.
For $R$ rounds, each node $u$ sends the next node $v_r$ from each color class to all neighbors $v_1,\ldots,v_{t_c}$ in that color class, and each neighbor $v_i$ responds by telling $u$ whether $v_r$ is its neighbor or not.
Node $u$ remembers this information; if at any point it knows of a subset $S \subseteq N_c(u)$ of $|S| = s$ nodes that are all neighbors of each other, then it has found an $s$-clique, and it rejects.
After $R$ rounds $u$ gives up and accepts.
 
\begin{lemma}\label{lem:kk_detection}
	If $G$ is $\epsilon$-far from $K_s$-free, then with probability at least $2/3$, at least one vertex detects a copy of $K_s$.
\end{lemma}
\begin{proof}
	As in lemma \ref{lem:triangle-detection}, consider a maximum set of edge-disjoint good $K_s$ copies in $G$, denoted $\mathcal{Q}$.
	Let $t = |\mathcal{Q}|$, $\mathcal{Q} = \set{H_1,\ldots,H_t}$,
	where $H_i = \set{v_i^1,\ldots,v_i^s}$ for each $i = 1,\ldots,t$.
	From corollary \ref{coro:good_copies} we know that
	\begin{equation*}
	t \geq \frac{\epsilon m}{4|E(K_s)|^2} \geq \frac{\epsilon m}{s^4}.
	\end{equation*}
	Assume w.l.o.g.\ that $v_i^1$ is a good vertex, for each $i = 1,\ldots,t$ (we know $H_i$ contains a good vertex, because it is a good copy).
	Let $X_i$ be an indicator for the event that for some color class $c$ we have $v_i^2,\ldots,v_i^s \in N_c(v_i^1)$,
	that is, node $v_i$ gave the same color to all other nodes of $H_i$.
	Then
	\begin{equation*}
	\Pr\left[ X_i = 1 \right] = \frac{1}{C(v_i^1)^{s-2}} = \frac{2s^4}{\eps m},
	\end{equation*}
	as the color assigned to each neighbor is independent of the others.
	Because $X_1,\ldots,X_t$ are independent, for their sum $X = \sum_{i = 1}^t X_i$ we have:
	\begin{align*}
	\Pr\left[ X = 0 \right]
	&=
	\Pr\left[ \bigcap_{i = 1}^t \left( X_i = 0 \right) \right]
	=
	\left( 1-\frac{2s^4}{\eps m}\right)^{t}
	\leq
	e^{-\frac{2s^4 t}{\eps m}} \leq  e^{-2}.
	\end{align*}
	
	For each $v_i^1$
	we have 
	$d(v_i^1) \leq \sqrt{2m|E(H)| /  \epsilon} = \sqrt{2m s(s-1)/\epsilon}$,
	because $v_i^1$ is a good vertex.
	Thus, for any color class $c$,
	given $X_i = 1$,
	the expected size of $N_c(v_i^1)$ is at most
	\begin{align*}
	&\frac{d(v_i) - (s-1)}{C(v_i)} + s - 1 \leq
	\frac{
		\sqrt{2m s(s-1)/\epsilon}
	}
	{
		\left\lceil{\left( \frac{1}{2s^4}\eps m\right) }^{\frac{1}{s-2}} \right\rceil
	}
	+ s - 1
	\\
	&
	\leq
	\sqrt{2} \cdot m^{1/2 - 1/(s-2)} \cdot \eps^{-1/2 - 1/(s-2)} \cdot s \cdot \left( 2s^4 \right)^{1/(s-2)} + s - 1
	< R / e^2.
	\end{align*}
	By Markov,
	\begin{align}
	\Pr\left[ \left| N_c(v_i^1) \right| > R \right]
	\leq
	1/e^2.
	\end{align}

	By union bound, the probability that $X = 0$, \emph{or} that the color class containing a good copy of $K_s$ is too large for the smallest $v_i^1$ with $X_i = 1$,
	is at most $1/e^2 + 1/e^2 < 1/3$.
\end{proof}

\begin{proof}[Proof of Theorem~\ref{thm:cliques}]
			
		If $G$ is $\eps$-far from being $K_k$-free, by lemma \ref{lem:kk_detection} then at least one vertex detects a copy of $K_s$ with probability at least $2/3$ and rejects. In the other hand, if $G$ is $K_s$-free, then clearly no vertex discovers a $K_s$, and all vertices accept.
\end{proof}
\begin{remark}
For $s \geq 5$, the algorithm requires a linear estimate of $m$ to get good running time. If $m$ is unknown, then the vertices may run the algorithm $\log n$ times for exponentially-increasing guesses $m = [n, 2n, ... n^2]$ ,  and as the protocol has one sided error, correctness is maintained; however, the running time increases to $O( \eps^{\frac{-s}{2(s-2)}}n^{\frac{s-4}{(s-2)}})$ rounds.
\end{remark}

\subsection{Constant-time algorithm for graphs with bounded maximal degree}
Finally, if the graph $G$ has maximum degree $\Delta = O((\eps m)^{\frac{1}{s-2}})$,
we
we can instantiate the algorithm with yet another setting for the number of color classes $C(u)$ and the timeout $R$,
to obtain a constant-time algorithm for testing $K_s$-freeness.
Note that as usual, we treat $s$ here as \emph{constant}, and we are interested only in the behavior with regard to $n, m$ and $\epsilon$.

	\begin{theorem}\label{thm:cliquesBoundedDeg}
		For any constant $s \geq 3$,
	there is a one-sided error property-testing algorithm for $K_s$ 
	for graphs with maximum degree $\Delta = O((\eps m)^{\frac{1}{s-2}})$,
	which runs in constant time (independent of $\epsilon$).
\end{theorem}
In particular, for the 5-clique we get:
\begin{corollary}
	Assuming maximal degree $\Delta = O(\sqrt[3]{\eps n})$, there is a one-sided error, $O(1)$-time distributed property-testing algorithm for $ K_5$-freeness.
\end{corollary}
This also extends to graphs with higher maximum degree, if their maximum and average degrees are of the same order of magnitude:
\begin{corollary}
	Assuming the maximal degree $\Delta = \Theta((\eps n)^{2/3})$ and average degree $\bar{d} = \Theta((\eps n)^{2/3})$, there is a one-sided error, $O(1)$-time  distributed property-testing algorithm for $ K_5$-freeness.
\end{corollary}

Assume $\Delta \leq (\alpha\eps m)^{\frac{1}{s-2}}$ for some constant $\alpha > 0$. 
Set
\begin{equation*}
	C(v) = \left\lceil \frac{d(v)}{(2\alpha)^{s-2}} \right\rceil
\end{equation*}
and
\begin{equation*}
	R = e^2((2\alpha)^{s-2} + (s - 1)).
\end{equation*}

This yields a constant-time algorithm, as $R$ is constant. We claim that if the graph is $\eps$-far from $K_s$-free, we will find a copy of $K_s$ with good probability.

\begin{proof}[Proof of Theorem \ref{thm:cliquesBoundedDeg}]
	Suppose that $G$ is $\eps$-far from $K_s$-free,
	and fix a set $\mathcal{Q}$ of $t \geq \eps m$ edge-disjoint copies of $K_s$.
	(This time, we do not require the copies to be good.)
	Let $v_i$ be some vertex from $H_i$, for each $i = 1,\ldots,t$,
	 and let $X_i$ indicate whether $v_i^1$ gave the same color to the other nodes of $H_i$.
	 For the sum $X = \sum_{i = 1}^t X_i$,
	 \begin{equation*}
		 \Pr\left[ X = 0 \right]
		 =
		 \Pr\left[ \bigcap_{i = 1}^t \left( X_i = 0 \right) \right]
		 =
		 \prod_{i = 1}^t \left( 1-\frac{2\alpha}{d(v_i)^{s-2}}\right)
		 \leq
		 \left( 1-\frac{2}{\eps m}\right)^{t}
		 \leq e^{-\frac{2t}{\eps m}}
		 \leq e^{-2}.
	 \end{equation*}

	 For the expected size of $N_c(v_i)$ given $X_i = 1$, we now get at most
	 \begin{align*}
		 \frac{d(v_i)}{C(v_i)} + s - 1
		 &\leq
		 \frac{ d(v_i) } { \left\lceil \frac{d(v_i)}{(2\alpha)^{s-2}} \right\rceil } + s - 1
		 \leq
		 (2\alpha)^{s-2} + s - 1
		 \leq R/e^2,
	 \end{align*}
	 so again the probability that the size of the relevant color class exceeds $R$ is at most $1/e^2$.

\end{proof}

\section{Towards Lower Bounds}

In this section we show that in some cases, some dependence on $\eps$ is necessary.
\subsection{\boldmath $\widetilde{\Omega}(1/\sqrt{\eps})$ lower bound on $C_5$}
\label{sec:lower}
In \cite{DFO14} it was shown that for sufficiently large $n$, there exists a
class of graphs over $n$ nodes,
with $m=\Theta(n^2)$ edges, on which solving \emph{exact} $C_5$-freeness (not the property-testing version)
requires $\widetilde{\Theta}(n)$ rounds.
If we instantiate this construction with $n = 1/\sqrt{\eps}$ nodes,
then whenever the graph contains a $5$-cycle, it is $\eps$-far from being $C_5$-free (a single edge corresponds to an $\eps$-fraction of edges,
since the total number of edges is $O(1/\eps)$).
Therefore we get:
\begin{observation}
	Any algorithm for testing $C_5$-freeness which does not depend on the size $n$ of the graph or the number of edges $m$
	requires $\widetilde{\Omega}(1/\sqrt{\eps})$ rounds.
	\label{obs:C5}
\end{observation}
(This can be extended to any odd-length cycle $C_k$ with $k \geq 5$.)

Interestingly,~\cite{DFO14} was \emph{not} able to prove a similar lower bound for exact triangle-freeness, and the problem remains open.
Since we have shown that triangle-freeness can be tested in $O(1)$ rounds when $\eps \geq \min \set{m^{-1/3}, n/m}$,
the technique of~\cite{DFO14} \emph{cannot} be extended to triangles, otherwise we would get an observation similar to Obs.~\ref{obs:C5}
for triangles,
which would be a contradiction.

\subsection{\boldmath A directed graph requiring $\widetilde{\Theta}(1/\eps)$ rounds}
The algorithms we gave in Sections~\ref{sec:cycles} and~\ref{sec:trees}
extend to the \emph{directed} \CONGEST model,
where each node knows only its incoming edges, and nodes communicate by broadcast (the broadcast is received by outgoing neighbors, but the sending node does not know who they are). 
This lets us test for directed $k$-cycles and trees oriented upwards towards the root.
We can show that in the directed \CONGEST model, there is a directed subgraph $H$ such that testing $H$-freeness requires $\Theta(1/\eps)$ rounds.

Consider the graph $H = ( V, E)$,
where $V = \set{0,1,2,3}$ and $E = \set{ (0,1), (0,2), (1,3), (2,3) }$.
We can test for $H$-freeness using color coding, by randomly choosing an edge to serve as $(0,1)$, and then using color-coded BFS to have find edges of $G$ matching the remaining edges of $H$; in the end, node $3$ knows if a copy was found or not.
If $G$ is $\eps$-far from $H$-free, then each attempt succeeds with probability $\geq \eps$, and the overall running time is $O(1/\eps)$.

An easy reduction from the Gap Disjointness problem in communication complexity shows that this is tight,
that is, $\Omega(1/(B\eps))$ rounds are required to test $H$-freeness in directed graphs,
where $B$ is the bound on the number of bits broadcast in each round.

In the Gap Disjointness problem, denoted $\prob{GapDisj}_{n,\eps}$, we have two players, Alice and Bob,
and they receive sets $X,Y \subseteq \set{1,\ldots,n}$, respectively.
Their goal is to distinguish the case where $|X \cap Y| = \emptyset$ from the case where $|X \cap Y| \geq \eps \cdot n$.
(If neither case holds, any output is allowed.)
It is known that to solve $\prob{GapDisj}_{n,\eps}$ the players must exchange $\Omega(1/\eps)$ bits of communication,
even if they can use randomization.
When $\eps < 1/2$, we may also assume that we never have $|X \cap Y| > n/2$; this does not make the problem easier.

The reduction from $\prob{GapDisj}_{n,\eps}$ to $H$-freeness is as follows.
Given inputs $X,Y$, Alice and Bob construct a graph $G_{X,Y}$, containing nodes $\set{ A, B, C_1, \ldots, C_5, D_1, \ldots, D_5, 1,\ldots, n}$.
The graph includes the following edges:
there is a path $A \rightarrow C_1 \rightarrow \ldots, C_5 \rightarrow B$ from $A$ to $B$ over nodes $C_1,\ldots, C_5$,
and another path $B \rightarrow D_1 \rightarrow \ldots \rightarrow D_5 \rightarrow A$ in the other direction using $D_1,\ldots,D_5$.
In addition, there are edges $(A, i)$ and $(B,i)$ for each $i \in \set{1,\ldots,n}$,
as well as an edge $(i, C_3)$.
So far, the graph does not contain any copies of $H$. Also, the graph is strongly connected.

Next, Alice and Bob examine $X$ and $Y$, and add the following edges: for each $i \in X$, Alice adds the edge $(i, A)$;
and for each $i \in Y$, Bob adds the edge $(i,B)$.
For each $i \in \set{1,\ldots,n}$, if we let $j \neq i$ be some other node in $\set{1,\ldots,n}$,
then a copy of $H$ over nodes $A, B, i, j$ iff $i \in X \cap Y$ (with node $A$ taking the role of 1, node $B$ taking the role of 2,
node $i$ taking the role of $0$, and node $j$ taking the role of $3$).
Thus, the graph is $\eps$-far from $H$-free iff $\prob{GapDisj}_{n,\eps}(X,Y) = 0$.

Alice and Bob can simulate the execution of a distributed algorithm in $G_{X,Y}$ as follows:
Alice simulates all the nodes except node $B$, and Bob simulates all the nodes except node $A$.
Both players use public randomness to generate the randomness of the nodes they simulate.
To simulate a round of the distributed algorithm,
Alice tells Bob the message sent by node $A$, and Bob tells Alice the message sent by node $B$. (The model has broadcast communication,
so each node broadcasts a single message.)
Next, Alice and Bob feed to each node they simulate the messages sent on all of its incoming edges.
In particular, because Alice knows $X$, she knows the incoming edges of node $A$, and similarly for Bob and node $B$.
The other nodes have a fixed set of incoming edges which does not depend on $X$ or $Y$.

The cost of the simulation is $2B$ bits per round, and since $\prob{GapDisj}_{n,\eps}$ requires a total of $\Omega(1/\eps)$
bits, the distributed algorithm for $H$-freeness must have $\Omega(1/(B\eps))$ rounds.



	\section{\boldmath Solving $K_s$ for $K_s$-Behrend graphs in $O(n^{o(1)})$ rounds}
	\label{sec:behrend}
Behrend graphs are a well studied family of graphs, and among their applications, they are used  in the world of classical property testing to show that testing triangle-freeness is hard in certain models.
An extension to these graphs for $K_5$-freeness was given in \cite{square-free}, and was used as a hard example for their algorithm. We show an algorithm that solves $K_s$-freeness on this family of graphs in $O(n^{o(1)})$ rounds, for any $s \geq 5$.
(We believe that more careful analysis of our algorithm may show that it only requires $O(1)$ rounds, and are currently working towards this.)
Our algorithm serves as evidence that Behrend graphs may \emph{not} be a hard example for $K_5$.
	
	In this section we show an algorithm that solves $K_s$-freeness on this family of graphs in $O(n^{o(1)})$ rounds.
	
	\subsection{Graph definition (Based on \cite{square-free})}
	
	\begin{lemma}[\cite{square-free} Lemma 2]\label{combinatorialBehrend}
		Let $k$ be a constant. For any sufficiently large $p$, there exists a set $X \subset \{0,..., p-1\}$ of size $p' \geq p^{1-\frac{\log\log{p}+4}{\log\log\log{p}}}$ such that, for any $k$ elements $x_1,x_2,..., x_k$ of $X$,
		$$ \sum_{i=1}^{k-1} x_i \equiv (k-1) x_k (\bmod p) \;\; \implies \;\;  x_1 = \dots = x_{k-1} = x_k.$$ 
	\end{lemma}
	
	\begin{construction}[$BC(s,n)$ \cite{square-free} Section 3]
		The graph $BC(s,n)$ is defined as follows: let $n$ be a prime, let $s$ an odd number, and let $V^1,...,V^s$ be sets, where $|V^i| = n$.
		Denote the $j$'th vertex of $V^i$ as $u_j^i$.
		Let $X$ be a the set from Lemma~\ref{combinatorialBehrend} with $k=s,p=n$.
		For all $x \in X$ and $i=1,...,s$, add the cycle $(u^1_i,u^2_{i+x (\bmod n)}...,u^s_{i +sx (\bmod n)})$ to the graph.  
	\end{construction}
	\begin{construction}[$BK(s,n)$ \cite{square-free} Section 3]
		The graph $BK(s,n)$ is defined as follows: let $n$ be a prime, let $s$ an odd number, and let $V^1,...,V^s$ be sets, where $|V^i| = n$.
		Denote the $j$'th vertex of $V^i$ as $u_j^i$. Let $X$ be a the set from Lemma~\ref{combinatorialBehrend} with $k=s,p=n$.
		For all $x \in X$ and $i=1,...,s$, add the edges of the clique $(u^1_i,u^2_{i+x(\bmod n)}...,u^s_{i +sx(\bmod n)})$ to the graph.    
	\end{construction}
	
	Clearly $BC(s,n)$ is a subgraph of $BK(s,n)$, where the edges remaining are between consecutive sets $V^i, V^{i + 1 \bmod p}$.
	The degree of each vertex in $BC(s,n)$ is exactly $n^{1-\frac{\log\log\log{n}+4}{\log\log{n}}}$.
	Denote $f(n) = n^{\frac{\log\log\log{n}+4}{\log\log{n}}}$.
	
	\subsection{Algorithm overview}
	The algorithm's key observations are as follows. 

	To start with, assume each vertex $u$ knows the vertex set $V^i$ to which it belongs (which is not true, but we will over come that later).
	Then we can find a copy of $K_s$ in $O(1)$ rounds, using a cycle-detection algorithm similar to the one in Section~\ref{sec:cycles}:
	if we consider $BC(s,n)$, the subgraph that contains only edges between consecutive vertex sets, then it is $\frac{1}{s}$-far from $C_s$-free;
	and from the construction we see that any cycle in $BC(s,n)$ supports an $s$-clique in $BK(s,n)$, so finding an $s$-cycle also means we have found an $s$-clique.

	
	We think of the vertex set $V_i$ to which node $u$ belongs as the \emph{color} of node $u$.
	It might not be possible to find the correct color for all the nodes,
	but because of the graph's high degree and structure,
	we can find a very large \emph{partial} coloring assigning many nodes $u$ to the correct vertex set $V_i$,
	and this is sufficient for the reduction to finding a cycle to go through.
	Under this partial coloring,
	when we consider only colored vertices and edges between consecutive vertex sets,
	we can show that any colored vertex has an $s$-cycle passing through it with high probability.
	Using a weighted color-coded BFS as in Section~\ref{sec:cycles} we can find one such cycle, and thereby find the $s$-clique supported on it.

	
	\subsection{Algorithm details}
	
	This partial coloring is attained by the following protocol.
	We obtain a large partial coloring as follows:
	for each $i = 1,\ldots,s$, we guess $s$ random vertices $v_i^1,\ldots,v_i^s$, and mark these nodes with the color $i$.
	In order to sample a random node, we have each node select itself with probability $1/n$;
	with constant probability, we get exactly one marked node for each $1 \leq i,j \leq s$, with no repetitions.
	Given this event, with constant probability, all $s^2$ marked vertices are colored correctly, that is, $v_i^j \in V^i$ for each $1 \leq i,j \leq s$.
	We condition on both events in the sequel.

	For each vertex set  $i=1,...,s$ the network guesses $s$ random vertices $j=1,...,s$ and marks these nodes with the color $i$. Sampling exactly a single random vertex for each $i,j$ can be simulated in constant probability by each vertex sampling itself with probability $\frac{1}{n}$ denoted $a$. For completeness we add a proof that this probability is at least $1-\frac{1}{e}-\frac{1}{2} > 0$ in lemma \ref{lem:const-prob-sample} at the end of this section.
	The probability that these $s^2$ vertices were marked with the correctly according to the vertex sets occurs with constant probability. 
	
	\begin{remark}
		From here on we condition that exactly one vertex was sampled in each $i,j$ vertices were sampled, and all the $s^2$ vertices chosen at random were colored correctly. This could be assumed due to the fact that the algorithm is $1$-sided by repeating the protocol $\left( \frac{s}{1-\frac{1}{e}-\frac{1}{2}}\right) ^{s^2}$ times it occurs with an arbitrary constant probability.
	\end{remark}
	For each $j\in[1,...,s]$, each vertex $v$ in the network maintains a set of colors $A_{j}(v) =\{1,...,s\}$. For $i=1,...,s$, $v$ considers whether it is connected to the $j$'th chosen vertex of color $i$, and if so is vertex removes $i$ from it's set of colors.
	\begin{definition}[Safe vertex]
		A vertex $v$ is a $j$-safe vertex if $|A_{j}(v)| = 1$, and if the single color in $A_{j}(v)$ is $j$.
	\end{definition}
	
	\begin{lemma}
		Conditioning that the algorithm guessed all the $s^2$ vertices colors correctly, let $v$ be a $j$-safe vertex, then $v \in V^j$.
	\end{lemma}
	\begin{proof}
		Conditioned on the assumption, since $v$ is $j$-safe, it has neighbors from all  vertex sets other than $V^j$, therefore it must be from $V^j$ 
	\end{proof}
	
	\begin{definition}[Safe $C_s$]
		A $C_s=(u_1,...,u_s)$ is defined as a safe cycle if for all $j$, it's $j$'th vertex is $j$-safe.
	\end{definition}
	\begin{lemma}
		Assuming the initial random vertices were picked correctly, if $c$ is a safe cycle, then all edges of $c$ are contained in the subgraph $BC(s,n)$.
	\end{lemma}
	\begin{proof}
		Assuming that the initial vertices were colored correctly, the $j$'th vertex is from $V^j$, meaning that only edges between consecutive layers (mod $s$) are considered. Therefore $c$ is in $BC(s,n)$.
	\end{proof}
	
	Color a vertex with color $j$ if it is $j$-safe. Consider the subgraph $G'$ that contains only the colored vertices, and only edges between two consecutive colors (mod $s$). 
	
	\begin{lemma}
	Let $v$ be a $1$-safe vertex. Denote $X_v$ to be the number of safe $C_s$'s passing through $v$. Then $Pr(X_v > 0) \geq \frac{1}{f(n)^{s(s-1)}}$.  
\end{lemma}
\begin{proof}
	From the construction of the graphs $BK(s,n)$ and $BC(s,n)$, each vertex in $BC(s,n)$ has exactly $\frac{n}{f(n)}$ $C_s$'s passing through it. Consider a cycle $c_0$ that passes through $v$. The $i$'th vertex of $c_0$ is $i$-safe with probability $\frac{1}{f(n)^{s-1}}$. This is because the degree of between vertex in any layer $V^l$ in $BK(s,n)$ to any other layer is exactly $\frac{n}{f(n)}$, and therefore in the $i$'th iteration the probability that all but the $i$'th color is removed is it's degree from each layer. Conditioning that $v$ is $1$-safe vertex,  the probability that $c_0$ is safe is  $\frac{1}{f(n)^{s(s-1)}}$ . This is due to the fact that each of the $s$ iterations that determine whether a vertex is $j$-safe for $j=1...s$ are independent.
	
	Consider a cycle $c \in BC(s,n)$ passing through $v$  The probability that this cycle is in $G'$ is at least $\frac{1}{f(n)^{s(s-1)}}$.
	
	
\end{proof}
	
	Given the partial coloring, similar to the cycle detection algorithm each vertex colored with color $1$ chooses a random weight from $[1,n^4]$, and proceeds to make a weighted priority BFS on the graph $G'$ for $s$ rounds. The weights are unique w.h.p, and the $1$ colored vertex with the maximal weight finishes its BFS uninterrupted.If this the $1$ colored vertices detected a cycle, reject and return it's vertices as the clique vertices, otherwise accept. 
	
	\begin{proof}[Proof of Correctness]
		By construction, if a $1$ vertex detects a cycle from the subgraph $BC(s,n)$, then it found a $K_s$ clique in the graph. The protocol is detects such a cycle assuming the initial sampling and coloring were correct, if the weights of the BFS are unique, and if there is a cycle from $BC(s,n)$ going through the maximal weighted $1$ vertex. Therefore the protocol succeeds with probability  $\frac{1-(1/e)-(1/2)}{s^{s^2}}(1-o(1))\frac{1}{f(n)^{s(s-1)}}$. Since the protocol's error is $1$-sided, the success probability could be amplified to any constant probability in $O(f(n)^{s(s-1)})=O(n^{o(1)})$ rounds.	
	\end{proof}

\begin{lemma}\label{lem:const-prob-sample}
	If each vertex samples itself with probability $\frac{1}{n}$, then with probability at least $1-\frac{1}{e}-\frac{1}{2} > 0$ a single vertex is sampled.
\end{lemma}
\begin{proof}
	Denote $X$ the number of vertices sampled. The probability that $X=0$ is exactly $(1-\frac{1}{n})^n \leq \frac{1}{e}$. Clearly $E[X] = 1$, therefore by Markov inequality $Pr(X \geq 2) \leq \frac{1}{2}$. Therefore $Pr(X = 1) \geq 1-\frac{1}{e}-\frac{1}{2} >0$
\end{proof}

\nocite{*}
\bibliographystyle{plainurl}
\bibliography{bibliography}

\begin{thebibliography}{10}

\bibitem{Alon02}
Noga Alon.
\newblock Testing subgraphs in large graphs.
\newblock {\em Random Struct. Algorithms}, 21(3-4):359--370, 2002.

\bibitem{AlonFKS00}
Noga Alon, Eldar Fischer, Michael Krivelevich, and Mario Szegedy.
\newblock Efficient testing of large graphs.
\newblock {\em Combinatorica}, 20(4):451--476, 2000.

\bibitem{AKKR08}
Noga Alon, Tali Kaufman, Michael Krivelevich, and Dana Ron.
\newblock Testing triangle-freeness in general graphs.
\newblock {\em {SIAM} J. Discrete Math.}, 22(2):786--819, 2008.

\bibitem{color_coding}
Noga Alon, Raphael Yuster, and Uri Zwick.
\newblock Color-coding.
\newblock {\em J. ACM}, 42(4):844--856, 1995.

\bibitem{Brakerski2011}
Zvika Brakerski and Boaz Patt-Shamir.
\newblock Distributed discovery of large near-cliques.
\newblock {\em Distributed Computing}, 24(2):79--89, 2011.

\bibitem{triangle_free}
Keren Censor-Hillel, Eldar Fischer, Gregory Schwartzman, and Yadu Vasudev.
\newblock {\em Fast Distributed Algorithms for Testing Graph Properties}, pages
  43--56.
\newblock 2016.

\bibitem{CHKKLPJ15}
Keren Censor{-}Hillel, Petteri Kaski, Janne~H. Korhonen, Christoph Lenzen, Ami
  Paz, and Jukka Suomela.
\newblock Algebraic methods in the congested clique.
\newblock In {\em Proceedings of the 2015 {ACM} Symposium on Principles of
  Distributed Computing, {PODC} 2015}, pages 143--152, 2015.

\bibitem{CGRSSS14}
Artur Czumaj, Oded Goldreich, Dana Ron, C.~Seshadhri, Asaf Shapira, and
  Christian Sohler.
\newblock Finding cycles and trees in sublinear time.
\newblock {\em Random Struct. Algorithms}, 45(2):139--184, 2014.

\bibitem{DLP12}
Danny Dolev, Christoph Lenzen, and Shir Peled.
\newblock {\em ``Tri, Tri Again'': Finding Triangles and Small Subgraphs in a
  Distributed Setting}, pages 195--209.
\newblock 2012.

\bibitem{DFO14}
Andrew Drucker, Fabian Kuhn, and Rotem Oshman.
\newblock On the power of the congested clique model.
\newblock In {\em Proceedings of the 2014 ACM Symposium on Principles of
  Distributed Computing}, PODC '14, pages 367--376, 2014.

\bibitem{Fischer01}
Eldar Fischer.
\newblock The art of uninformed decisions.
\newblock {\em Bulletin of the {EATCS}}, 75:97, 2001.

\bibitem{square-free}
Pierre Fraigniaud, Ivan Rapaport, Ville Salo, and Ioan Todinca.
\newblock {\em Distributed Testing of Excluded Subgraphs}, pages 342--356.
\newblock 2016.

\bibitem{G16}
Fran{\c{c}}ois~Le Gall.
\newblock Further algebraic algorithms in the congested clique model and
  applications to graph-theoretic problems.
\newblock In {\em Distributed Computing - 30th International Symposium, {DISC}
  2016}, pages 57--70, 2016.

\bibitem{Gol98}
Oded Goldreich.
\newblock Combinatorial property testing -- a survey.
\newblock {\em Randomization Methods in Algorithm Design}, 1998.

\bibitem{GGR98}
Oded Goldreich, Shari Goldwasser, and Dana Ron.
\newblock Property testing and its connection to learning and approximation.
\newblock {\em J. ACM}, 45(4):653--750, July 1998.

\bibitem{HP16}
Zengfeng Huang and Pan Peng.
\newblock Dynamic graph stream algorithms in o(n) space.
\newblock In {\em 43rd International Colloquium on Automata, Languages, and
  Programming, {ICALP} 2016, July 11-15, 2016, Rome, Italy}, pages 18:1--18:16,
  2016.

\bibitem{Ron09}
Dana Ron.
\newblock Algorithmic and analysis techniques in property testing.
\newblock {\em Foundations and Trends in Theoretical Computer Science},
  5(2):73--205, 2009.

\end{thebibliography}

\end{document}